\documentclass[a4paper,10pt]{article}
\usepackage[utf8]{inputenc}
\usepackage{authblk}
\usepackage{amsmath, amsfonts, amssymb, amsthm}
\usepackage{graphicx}
\usepackage{appendix}
\usepackage{tabularx}
\usepackage{diagbox}
\usepackage{cancel}
\usepackage{lscape}
\usepackage{color}
\usepackage{mathtools}
\usepackage{enumitem}
\usepackage{bbm}

\DeclareMathOperator{\tr}{tr}

\DeclareMathOperator{\dev}{dev}
\DeclareMathOperator{\sym}{sym}
\DeclareMathOperator{\SO}{\mathrm{SO}}

\DeclareMathOperator{\Curl}{\mathrm{Curl}}
\DeclareMathOperator{\Grad}{\mathrm{Grad}}
\DeclareMathOperator{\curl}{\mathrm{curl}}

\DeclareMathOperator{\vol}{\mathrm{vol}}
\DeclareMathOperator{\Adj}{\mathrm{Adj}}
\DeclareMathOperator{\axl}{\mathrm{axl}}
\DeclareMathOperator{\Axl}{\mathrm{Axl}}
\DeclareMathOperator{\anti}{\mathrm{anti}}
\DeclareMathOperator{\Anti}{\mathrm{Anti}}

\newcommand{\normal}{\nu}
\newcommand{\norm}[1]{\left|\!\left|#1\right|\!\right|}
\newcommand{\R}{\mathbb{R}}

\newcommand{\inv}{^{-1}}
\newcommand{\id}{{\boldsymbol{\mathbbm{1}}}}

\newcommand{\Rdd}{\R^{3\times 3}}
\newcommand{\sprod}[1]{\langle#1\rangle}
\newcommand{\matr}[1]{\begin{pmatrix}#1\end{pmatrix}}
\newcommand{\oR}{\overline{R}}
\newcommand{\oA}{\overline{A}}
\newcommand{\oU}{\overline{U}}

\newcommand{\oK}{\overline{K}}
\newcommand{\fraK}{\mathfrak{K}}

\newcommand{\ein}{\,\rule[-3pt]{0.4pt}{8pt}\,{}}

\newcommand{\phii}{\varphi}
\newcommand{\Om}{\Omega}
\newcommand{\col}{\colon}
\newcommand{\sub}{\subset}
\newcommand{\Rtilde}{\widetilde{R}}

\newtheorem{theorem}{Theorem}

\newtheorem{lemma}[theorem]{Lemma}
\newtheorem{definition}[theorem]{Definition}
\newtheorem{consequence}{Consequence}
\newtheorem{remark}[theorem]{Remark}

\theoremstyle{remark}

\begin{document}
\title{Integrability conditions between the first and second Cosserat deformation tensor in geometrically nonlinear micropolar models and existence of minimizers}
\author{Johannes Lankeit\,%
\thanks{Johannes Lankeit, Institut f\"ur Mathematik, Universit\"at Paderborn,  Warburger Str. 100,
33098 Paderborn, Germany, email: johannes.lankeit@math.uni-paderborn.de}\,\,,\,
Patrizio Neff\,%
\thanks{Corresponding author: Patrizio Neff, Head of Lehrstuhl f\"ur Nichtlineare Analysis und Modellierung, Fakult\"at f\"ur  Mathematik, Universit\"at Duisburg-Essen, Campus Essen, Thea-Leymann Str. 9, 45127 Essen, Germany, email: patrizio.neff@uni-due.de, Tel.: +49-201-183-4243}\,\,,\,
Frank Osterbrink\,%
 \thanks{\,Frank Osterbrink, Fakult\"at f\"ur  Mathematik, Universit\"at Duisburg-Essen, Campus Essen, Thea-Leymann Str. 9, 45127 Essen, Germany, email: frank.osterbrink@uni-due.de}
}


\maketitle
\begin{center}
 {\it In memory of Claude Vall{\'e}e ($\dagger$\!\! November 2014)}
\end{center}

\begin{abstract}
In this note we extend integrability conditions for the symmetric stretch tensor $U$ in the polar decomposition of the deformation gradient $\nabla\varphi=F=R\,U$ to the non-symmetric case. In doing so we recover integrability conditions for the first Cosserat deformation tensor. Let $F=\oR\,\oU$ with $\oR:\Omega\subset\R^3\longrightarrow\SO(3)$ and $\oU:\Omega\subset\R^3\longrightarrow \mathrm{GL}(3)$. Then 
\begin{align*}
	\fraK:=\oR^T\mathrm{Grad}\,\oR=\Anti\left(\frac{\Big[\oU(\Curl \oU)^T-\frac{1}{2}\tr(\oU(\Curl \oU)^T)\id\Big]\oU}{\det \oU}\right),
\end{align*}
 giving a connection between the first Cosserat deformation tensor $\oU$ and the second Cosserat tensor ${\fraK}$. (Here, $\Anti$ denotes an isomorphism between $\R^{3\times 3}$ and $\mathfrak{So}(3):=\{\,\mathfrak{A}\in\R^{3\times 3\times 3}\,|\,\mathfrak{A}.u\in\mathfrak{so}(3)\;\forall u\in \R^3\}$.) The formula shows that it is not possible to prescribe $\oU$ and $\fraK$ independent from each other.

We also propose a new energy formulation of geometrically nonlinear Cosserat models which completely separate the effects of nonsymmetric straining and curvature. For very weak constitutive assumptions (no direct boundary condition on rotations, zero Cosserat couple modulus, quadratic curvature energy) we show existence of minimizers in Sobolev-spaces.
\end{abstract}


\textbf{Mathematics Subject Classification (MSC2010):} 74A35, 74A30\\

\section{Introduction}\label{sec1}
Classical continuum mechanics considers material continua as simple point-continua with points having three displacement-degrees of freedom, and the response of a material to the displacement of its points is characterized by a symmetric Cauchy stress tensor presupposing that the transmission of loads through surface elements is uniquely determined by a force vector, neglecting couples. But such a model is insufficient for the description of deformable solids with a certain microstructure, such as cellular materials, foam-like structures (e.g., bones), periodic lattices. That is the point where the geometrically nonlinear Cosserat model comes in.\\

In order to unify field theories embracing mechanics, optics and electro-dynamics by means of a common principle of least action (Euclidean action), the brothers Francois and Eugene Cosserat tried to find the correct general form of the energy for a deformable solid (see \cite{Cosserat09}). In contrast to the classical theory any material point in a Cosserat continuum (also called \emph{micropolar} continuum) is allowed to rotate without stretch, like an infinitesimal rigid body. Thus, in addition to the deformation field ${\varphi}$ an independent rotation field $ \,\overline{{R}}$ is needed to fully describe the generalized continuum. The Cosserat model has therefore 6 degrees of freedom (3 for the displacements and 3 for the rotations of points), in contrast to the classical elastic continuum, which has only 3 degrees of freedom.\\

From a historical perspective, the kinematical model of such generalized continua was introduced by the Cosserat brothers in \cite{Cosserat09}, but they did not provide any constitutive relations. Fifty years later, Ericksen and Truesdell \cite{Ericksen-Truesdell-58} have developed this idea and have drawn anew the attention to the theory of generalized continua. In the 1960s, several variants of the theory of media with microstructure were proposed by Toupin \cite{Toupin64}, Mindlin \cite{Mindlin64}, Eringen \cite{Eringen67} and others. 
The micropolar (or Cosserat) continua can be viewed as a special case of the microstretch continua and of the micromorphic continua (see Eringen \cite{Eringen98} for the foundation of this theory). In the last decades, the generalized theories of continua have proven useful for the treatment of complex mechanical problems, for which the classical theory of elasticity is not satisfactory. In this respect, we mention for instance the works of Lakes \cite{Lakes93,Lakes93b,Lakes98}, Neff \cite{Neff_Cosserat_plasticity05}, Neff and Forest \cite{Neff_Forest07}, among others.\\

The mathematical problem related to the deformation of Cosserat elastic bodies has been investigated in many works. In the linearized micropolar elasticity the existence and the properties of solutions have been studied by Ie\c san \cite{Iesan70,Iesan71,Iesan82} and by Neff \cite{Neff_zamm06,Neff_JeongMMS08} under some weaker assumptions, among others. For the geometrically nonlinear micromorphic model, the first existence theorem based on convexity arguments was presented by Neff \cite{Neff_Edinb06,Neff-PAMM-2004,Neff_Habil04}. Another existence result for generalized continua with microstructure was proved in \cite{Mariano08}, under certain convexity assumptions. Adapting the methods from \cite{Ball77,Dacorogna89}, Tamba\v ca and Vel\v ci\'c \cite{Tambaca10} have derived an existence theorem for nonlinear micropolar elasticity, for a general constitutive behaviour.\\

In our paper we consider a new energy formulation of geometrically nonlinear Cosserat models which completely separates the effects of nonsymmetric straining and curvature. 

In Section \ref{sec2} we introduce the strain measures and curvature strain measures associated to the classical nonlinear Cosserat model. The next section presents the expression of the elastically stored energy density and recalls the main assumptions on the constitutive coefficients.\\

In Section \ref{sec4} we provide an easy and self contained approach towards compatibility conditions in micropolar media. Starting point is Vall{\'e}e's formula (Theorem \ref{th:comp1}) providing compatibility relations generated by the classical polar decomposition. As it turns out, this formula can be generalized to the situation $F=\oR\cdot\oU$, where $\oR$ is still orthogonal, but $\oU$ need not be positive definite and symmetric. Our final compatibility result  is already known. Indeed, Teresi and Tiero (cf. \cite{Teresi97}, eq. (22)) use a much more abstract setting with an extremely short algebraic proof as compared to the classical, componentwise approach (cf. \cite{Kafadar71}). Their result seems not to be very well known and it has not been immediately connected to the classical compatibility result for the symmetric stretch tensor. In our presentation we choose a compact direct tensor notation which is a middle between the two extreme cases above. 

The advantage of Teresi and Tiero's analysis is, however, to make the structure more transparent and to highlight the appearance of the curl of the first Cosserat deformation tensor. This dependence is not easily grasped in the compatibility conditions given in the classical approach (see \cite{Kafadar71} and Eringen's book \cite{Eringen98}, eq. 1.7.4).\\

Having proven these relations between dislocation density $\oK$, strain measure $\oU$ and its derivative in the form of $\Curl\oU$, in Section \ref{sec5} we finally give the new energy formulation and briefly outline a proof for the existence of minimizers.

\subsection{Notation}
In this paper, we denote by $\R^{3\times 3}$ the set of real $3\times 3$ second order tensors, written with
capital letters and the Einstein summation convention over repeated indices is used. For ${a},{b}\in\R^3$ we let $\langle {a},{b}\rangle$ denote the scalar product on $\R^3$ with associated vector norm $\|{a}\|^2=\langle {a},{a}\rangle$. Similarly, the standard Euclidean scalar product on $\R^{3\times 3}$ is given by $\langle X,{Y}\rangle=\tr({X {Y}^T})$, and thus the Frobenius tensor norm is $\|X\|^2=\langle X,X\rangle$. 
The identity tensor on $\R^{3\times 3}$ will be denoted by $\id$, so that $\tr X=\langle X,{\id}\rangle$. We adopt the usual abbreviations of Lie-algebra theory, i.e.,
 $\mathfrak{so}(3):=\{X\in\mathbb{R}^{3\times3}\;|X^T=-X\}$ is the Lie-algebra of  skew symmetric tensors
and $\mathfrak{sl}(3):=\{X\in\mathbb{R}^{3\times3}\;| \tr({X})=0\}$ is the Lie-algebra of traceless tensors. 
$\mathrm{Sym}(3)=\{x\in\Rdd|\,X=X^T\}$ denotes the set of symmetric matrices. 
For all vectors ${\xi},{\eta}\in\R^3$ we have the tensor product $({\xi}\otimes{\eta})_{ij}={\xi}_i\,{\eta}_j$  and ${\epsilon}$ is the Levi-Civita tensor, also called the permutation tensor or third order alternator tensor, given by
\begin{align}
\epsilon_{ijk}=\left\{
\begin{array}{ll}
1& \text{if} \quad (i,j,k) \quad \text{ is an even permutation of} \quad (1,2,3)\\
-1& \text{if} \quad (i,j,k) \quad \text{ is an odd permutation of} \quad (1,2,3)\\
0 & \text{otherwise}.
\end{array}\right.
\end{align}
For all $X\in\mathbb{R}^{3\times 3}$ we set $\sym\,X=\frac{1}{2}(X^T+X)\in\mathrm{Sym}(3)$, 
$\mathrm{skew}\, X=\frac{1}{2}(X-X^T)\in \mathfrak{so}(3)$ and the deviatoric part $\dev X=X-\frac{1}{3}(\tr{X})\!\cdot\!\id\in \mathfrak{sl}(3)$  and we have the \emph{orthogonal Cartan-decomposition of the Lie-algebra} $\mathfrak{gl}(3)$
\begin{align}
	\begin{split}
		\mathfrak{gl}(3)&=\big\{\,\mathfrak{sl}(3)\cap \mathrm{Sym}(3)\,\big\}\,\oplus\,\mathfrak{so}(3)\,\oplus\,\mathbb{R}\!\cdot\! \id\,,\\
		X&=\dev \sym X+ \mathrm{skew} X+\frac{1}{3}\big(\tr X\big)\!\cdot\! \id\,.
	\end{split}
\end{align}

Here, for
\begin{align}
A=\left(\begin{array}{ccc}
0 &-a_3&a_2\\
a_3&0& -a_1\\
-a_2& a_1&0
\end{array}\right)\in \mathfrak{so}(3)
\end{align}
we consider the operators $\mathrm{axl}:\mathfrak{so}(3)\rightarrow\mathbb{R}^3$ and $\mathrm{anti}:\mathbb{R}^3\rightarrow \mathfrak{so}(3)$ given by
\begin{align}\label{eq:not_anti}
\axl\left(\begin{array}{ccc}
0 &-a_3&a_2\\
a_3&0& -a_1\\
-a_2& a_1&0
\end{array}\right):=\left(\begin{array}{c}
a_1\\
a_2\\
a_3
\end{array}\right),\quad \quad A\cdot v=(\axl A)\times v, \quad \quad \mbox{for all } \, v\in\mathbb{R}^3,\notag\\[-12pt] \\\quad A_{ij}=-\epsilon_{ijk}(\axl A)_k=:(\anti(\axl A))_{ij}, \quad \quad (\axl A)_k=-\frac{1}{2} \epsilon_{ijk}A_{ij}\,.\notag
\end{align}
Furthermore the axial operator $\mathrm{Axl}:\mathfrak{So}(3)\longrightarrow\R^{3\times 3}$ associated to an $\mathfrak{so}(3)$-valued third order tensor ${\mathfrak{A}}$ is defined by 
\begin{align}
	\big(\mathrm{Axl}\,{\mathfrak{A}}.u\big)\times v=\big(\mathfrak{A}.u\big)v,\qquad u,v\in\R^3\,,
\end{align}
where $\mathfrak{So}(3):=\{\,\mathfrak{A}\in\R^{3\times 3\times 3}\,|\,\mathfrak{A}.u\in\mathfrak{so}(3)\}$ and we define $\mathrm{Anti}:\R^{3\times 3}\longrightarrow\mathfrak{So}(3)$ to be
\begin{align}\label{eq:defAnti}
	\mathfrak{S}_{ijk}=-\epsilon_{ijl}\big(\mathrm{Axl}\,\mathfrak{S})_{lk}=:\big(\mathrm{Anti}\,\big(\mathrm{Axl}\,\mathfrak{S}\big)\big)_{ijk}.
\end{align}
Then, for $A\in \R^{3\times 3}$ we have
\begin{align*}
 \epsilon:\Anti A &=\big(\epsilon_{ijk}\,e_i\otimes e_j\otimes e_k\big):\big(-\epsilon_{rst}A_{tu}\,e_r\otimes e_s\otimes e_u\big)
 		     =-\epsilon_{ijk}\epsilon_{rst}A_{tu}\delta_{jr}\delta_{ks}\,e_i\otimes e_u\\
 		     &=-\epsilon_{ijk}\epsilon_{jkt}A_{tu}\,e_i\otimes e_u
 		     =-2\,\delta_{it}A_{tu}\,e_i\otimes e_u=2\,A_{iu}\,e_i\otimes e_u=-2A
\end{align*}
and thus, applying this to $A= \mathrm{Axl}\,\mathfrak{A}$, 
\begin{align}\label{eq:Axl}
	(\mathrm{Axl}\,\mathfrak{A})_{ij}=-\frac{1}{2}(\epsilon:\mathfrak{A})_{ij}=-\frac{1}{2}\epsilon_{irs}\,\mathfrak{A}_{rsj}. 
\end{align}
Here, the double dot product `` : '' of two third order tensors ${A}=A_{ijk}\,{e}_i\otimes{e}_j\otimes {e}_k$ and ${B}=B_{ijk}\,{e}_i\,\otimes{e}_j\otimes {e}_k$ is defined as $\,{A}:{B}\,=\, A_{irs}B_{rsj}\,{e}_i\otimes{e}_j\,$.
The cross product ``$\times$'' of two second order tensors is calculated with help of the rule $\,({a}\,\otimes\,{b})\times({c}\,\otimes\,{d})= {a}\,\otimes\,({b}\times{c})\,\otimes\,{d}\,$, which holds for any vectors ${a},{b},{c}$ and ${d}$.\\
Throughout this paper (when we do not specify else) Latin subscripts take the values $1,2,3$. Typical conventions for differential operations are implied such as comma followed by a subscript to denote the partial derivative with respect to the corresponding cartesian coordinate (e.g., $f,_i=\frac{\partial f}{\partial x_i}\,$).

\section{Strain and curvature measures for the classical Cosserat model}\label{sec2}
We consider an elastic body ${\mathcal{B}}=\{\mathcal{P}_k\}$ as a set of particles $\mathcal{P}_k$, which in the reference configuration occupies a bounded domain $\Omega\subset\mathbb{R}^3$ with Lipschitz boundary $\partial\Omega$. Let $ Ox_1x_2x_3$ be a Cartesian coordinate frame in $\mathbb{R}^3$ and ${e}_i$ the unit vectors along the coordinate axes $Ox_i\,$. Then any material point $\mathcal{M}$ in $\mathcal{P}_k$ may be identified by the spatial location $\varphi(x)$ of the centroid $x=(x_1,x_2,x_3)$ of $\mathcal{P}_k$ and its relative position vector ${\xi}(x)$. Furthermore the behaviour of the body under the action of some external loads can be described by the \emph{deformation field} 
\[
\varphi\colon\Omega\longrightarrow\R^3\,,\qquad \varphi(x)=\varphi(x_1,x_2,x_3)
\]
and a second order tensor field $\oR$ (also called \emph{microrotation tensor}) according to
\[
\oR\colon\Omega\longrightarrow\mathrm{SO}(3)\,,\quad {\xi}(x)=\oR(x_1,x_2,x_3)\,{\xi_0}
\]
with some fixed $\xi_0\in\R^3$.\\
We identify the second order tensor $\oR$ with the $3\times 3$ matrix of its components $\,\oR\in \mathbb{R}^{3\times 3}$ and the vector $\,{\varphi}\,$ with the column-vector of its components $\,\varphi\in\mathbb{R}^3$. Notice that $\oR$ is a proper orthogonal tensor, i.e. $\oR\in\SO(3)$. The three columns of the matrix $\oR$ will be denoted by $\,{d}_1\,,\,{d}_2\,,\,{d}_3\,$. They are usually called the \emph{directors} and can be interpreted as an orthonormal triad of vectors $\{{d}_i\}$ rigidly attached to each material point, describing thus the microrotations. Using either the direct tensor notation or the matrix notation, we can write respectively
\[\oR={d}_i\otimes {e}_i\qquad\mathrm{or}\qquad\oR=
\big({d}_1\,|\,{d}_2\,|\,{d}_3\,\big)_{3\times 3}\,\,.\]
In the reference configuration $\Omega$ the directors attached to every material point $(x_1,x_2,x_3)$ are taken to be $\{\, {e}_i\,\}$ and after deformation they become $\{{d}_i(x_1,x_2,x_3)\}$. If we denote by $\,\varphi_0(x_1,x_2,x_3):=x_i\,{e}_i$ the position vector of points in the reference configuration $\Omega$, then the \emph{displacement vector}  field is $ {u}=\varphi-\varphi_0\,$.\\

To introduce nonlinear strain measures let us now consider a material point at position $x\in{\cal P}_k$ and let $x+h$ be a neighboring material point within ${\cal P}_k$. Then we get
\[ 
{\varphi}(x+h)-{\varphi}(x)=\mathrm{Grad}\,\varphi(x)\,h+{o}(h),
\] 
which provides us with a primary measure of deformation, the so called \emph{deformation gradient}
\begin{align}\label{e1}
	F\,=\,\nabla{\varphi}=\mathrm{grad}\,{\varphi}= {\varphi},_i \otimes \, {e}_i = \big(\,{\varphi},_1\,| \, \,{\varphi},_2\,|\,\,{\varphi},_3\,\big)_{3\times 3}\,.
\end{align}

With this the Cosserat-brothers discussed an energy functional of the form
\begin{align}\label{e2}
\begin{split}
   I(\varphi, \oR) &= \int_\Omega W(\,F,\,\oR,\,\mathrm{Grad}\,\oR\,)\, \mathrm{d}V -\Pi(\,\varphi,\,\oR\,),\text{ where }\\
    \Pi(\,\varphi,\, \oR\,) &= the\;\emph{external loading potential} 
\end{split}
\end{align}
and postulated invariance under Euclidean transformations (today also known as the \emph{principle of frame-invariance}), which states that the energy density has to be left invariant under rigid rotations
\[W(\,QF,\,Q\oR,\,Q\,\mathrm{Grad}\,\oR\,)=W(\,F,\,\oR,\,\mathrm{Grad}\,\oR\,)\qquad\mbox{for all }\;{Q\in\SO(3)}\,.\]
With $Q=\oR^T$ we get the reduced representation (cf. \cite{Cosserat09,AltEre_13})
\[\,W(\,F,\,\oR,\mathrm{Grad}\,\oR\,)=W(\,\oR^TF,\,\oR^T\oR,\,\oR^T\mathrm{Grad}\,\oR\,)=W(\,\oR^TF,\oR^T\mathrm{Grad}\,\oR\,)\,,\] 
which motivates to consider \textbf{the first Cosserat deformation tensor} (also called \emph{the nonsymmetric right stretch tensor}, see the original Cosserat book\cite{Cosserat09}, p. 123, eq. (43))
\begin{align}\label{e3}
\oU=\oR^TF=\, ({e}_i \otimes {d}_i)({\varphi},_j \otimes {e}_j)= \langle\,{d}_i\,,{\varphi},_j\rangle\, {e}_i \otimes {e}_j\,.
\end{align}
and the \textbf{second Cosserat deformation tensor} (also called \emph{the third order curvature tensor}, see the original book \cite{Cosserat09}, p. 123, eq. (44))
\begin{align}\label{e4}
	\begin{split}
 		   \fraK\,&=\,\oR^T\mathrm{Grad}\,\oR\,=\,\oR^T( \oR,_k \otimes \, {e}_k) \,=\,\big(\oR^T \oR,_k \big)\otimes \, {e}_k\\[4pt]
    			     \,&=\,\langle\,{d}_i\,,{d}_{j,k}\,\rangle \, {e}_i \otimes {e}_j \otimes\,  {e}_k\,=\,\big(\, \overline{{R}}^T\overline{{R}},_1\,| \,\, \overline{{R}}^T \overline{{R}},_2\,|\,\,\overline{{R}}^T \overline{{R}},_3\,\big)\,,
	\end{split}
\end{align}
as proper deformation measures. Although $\, {\mathfrak{K}} \,$ is a third order tensor, it has, in fact, only 9 independent components, since $\,\overline{{R}}^T\overline{{R}},_k\,$ is skew-symmetric ($k=1,2,3$).

\begin{remark}
	We write the above tensors $\oR$ and $\oU$ with superposed bars in order to distinguish them from the factors  $R$ and $U$ of the classical polar decomposition $F=R\,U$, in which $R$ is orthogonal and $U$ is positive definite, symmetric and which is a standard notation in elasticity.
\end{remark}

Then, according to the classical Cosserat theory, we define the relative Lagrangian \emph{strain measure for stretch} by
\begin{equation}\label{e5}
    \overline{{E}}\,=\,\overline{{U}}- \id_3\,=\,\big(\,\langle\,{d}_i\,,{\varphi},_j\rangle-\delta_{ij}\,\big)\, {e}_i \otimes {e}_j\,,
\end{equation}
where $\,\delta_{ij} $ is the Kronecker symbol and $\,\id_3:={e}_i \otimes {e}_i$ is the unit tensor in the 3-space and for the Lagrangian strain measure for orientation change (curvature) we use the second Cosserat deformation tensor.\\

If we take the transpose in the last two components of $\fraK$, then we obtain the third order curvature tensor ${\widetilde{\mathfrak{K}}}$ which was used in \cite{Neff_plate04_cmt,Neff_Edinb06,Neff_Forest07}:
\begin{align}\label{e6}
	\begin{split}
    {\widetilde{\mathfrak{K}}}\,&=\,\fraK^{\stackrel{2.3}{T}} \,=\,\langle\,{d}_i\,,{d}_{k,j}\,\rangle \, {e}_i \otimes {e}_j \otimes\,  {e}_k\,=\, ({e}_i \otimes {d}_i) \big({d}_{k,j} \otimes {e}_j \big)\otimes \, {e}_k \vspace{4pt}\\
    \,&=\,\overline{{R}}^T\big(\, \mathrm{Grad}\, {d }_k \big) \otimes \, {e}_k \,=\, \big(\, \overline{{R}}^T\mathrm{Grad}\, {d }_1\,| \,\, \overline{{R}}^T \mathrm{Grad}\, {d }_2\,|\,\,\overline{{R}}^T \mathrm{Grad}\, {d }_3\,\big)\,.
 \end{split}
 \end{align}
In order to avoid working with a third order tensor for the curvature, one can replace $\,{\mathfrak{K}}\, $ by a second order curvature tensor. This can be done in several ways. Indeed, a first way is to consider the axial vector of the skew-symmetric factor $\oR^T \oR,_k$ in the definition \eqref{e4} of $\fraK$ and to introduce thus the second order tensor
\begin{align}\label{e7}
     \Gamma\,=\,\mathrm{axl}\big(\oR^T \oR,_k \big)\otimes \, {e}_k\,\,.
 \end{align}
The tensor $\Gamma$ is a Lagrangian measure for curvature and it is frequently called the \emph{wryness tensor} in the literature (see e.g., \cite{Pietraszkiewicz09}). The relation between $\Gamma$  and $\fraK$ can be expressed with the help of the third order alternator tensor
\[
{\epsilon}=-\id_3\times \id_3=\epsilon_{ijk}\,{e}_i\otimes{e}_j\otimes {e}_k
\]
in the form (since $\,\mathrm{axl}\,{S}=-\frac{1}{2}\,\,{\epsilon}:{S}\,$, for any skew-symmetric second order tensor ${S}$)
\begin{equation}\label{e8}
    \Gamma\,=\,-\,\dfrac{1}{2}\,\,{\epsilon}:\fraK \quad\qquad\mathrm{and}\qquad\quad \fraK\,=\, \id_3\times \Gamma\,=\, -\, {\epsilon}\,\Gamma\,.
\end{equation}

As an alternative to the wryness tensor $\,{\Gamma}\,$ given by \eqref{e4}, one can also use the $\Curl$ operator instead of $\mathrm{Grad}$ in the definition \eqref{e4} of $\fraK$. In doing so, we define the \emph{dislocation density tensor} $\oK$ by
\begin{align}\label{e9}
    \oK\,=\, \oR^T\Curl\, \oR\,=\, -\,\oR^T\big(\oR,_i \times \, {e}_i\big) \,\,.
\end{align}
The dislocation density tensor $\,\overline{{K}}\,$ presents some advantages for micropolar and micromorphic media, as it can be seen in \cite{Neff_curl08,NeffQuart,NeffCMT,MadeoNeffGhibaW}.
Here, the $\,\Curl\,$ operator for vector fields $\,{v}=v_i\, {e}_i$ has the well-known expression
\[
\curl\,{v}= \epsilon_{ijk}\,v_{j,i}\,{e}_k= -{v},_i \times {e}_i\,\,,
\]
while the $\,\Curl\,$ operator for tensor fields $\,{P}=P_{ij}\,{e}_i\otimes {e}_j\,$ is defined as
\begin{align}\label{e10}
    \Curl\,{P} \,=\,\epsilon_{jrs}\,P_{is,r}\, {e}_i\otimes {e}_j\,=\,-\,{P},_i \times\, {e}_i\,\,,
\end{align}
In other words, $\,\Curl\,$ is defined row wise as in \cite{Mielke06,Svendsen02}: the rows of the $3\times 3$ matrix $\,\Curl\,{P}\,$ are the 3 vectors $\,\curl\,{P}_i\,$ ($i=1,2,3$), respectively. 
Note that sometimes other definitions of $\Curl$ are used, for example in \cite{Vallee92}. 
The Curl operator there, which we will denote by $\widetilde{\Curl}$, is defined by the requirement $(\widetilde{\Curl}\, P)u = \curl (Pu)$, entailing that $\widetilde{\Curl} (P^T) = (\Curl P)^T$. \\

We describe next the close relationship between the wryness tensor $\Gamma$, the dislocation density tensor $\oK$ and the second Cosserat deformation tensor. One can prove by a straightforward calculation that the following relations hold \cite{Neff_curl08}
\begin{align}\label{e11}
    -\,\oK\,=\,\Gamma^T-(\mathrm{tr}\,\Gamma)\, \id_3\qquad\quad \mathrm{and}\qquad\quad -\,\Gamma\,=\,\oK^T-\dfrac{1}{2}\, (\mathrm{tr}\,\oK)\, \id_3,
\end{align}
as with \eqref{e8}
\begin{align}\label{e12}
	\fraK\,=\,-{\epsilon}\,\Gamma\,=\,{\epsilon}\,\oK^T-\frac{1}{2}\big(\mathrm{tr}\,\oK\,\big)\,{\epsilon}.
\end{align}
For infinitesimal strains this formula is well-known under the name \emph{Nye's formula}, see \cite{Neff_curl08}, and there $(-\Gamma)$ is also called \emph{Nye's curvature tensor} (cf. \cite{Nye53}). From \eqref{e11} we deduce the following relations between the traces, symmetric parts and skew-symmetric parts of these two tensors
\begin{align}\label{e15}
    \mathrm{tr}\,\oK\,=\, 2\,\mathrm{tr}\,\Gamma,\qquad \mathrm{skew}\,\oK\,=\, \mathrm{skew}\,\Gamma,\qquad \mathrm{dev\,sym}\,\oK\,=\,-\, \mathrm{dev\,sym}\,\Gamma.
\end{align}

In view of \eqref{e11}--\eqref{e15}, we see that one can work either with the wryness tensor $\,{\Gamma}\,$ or, alternatively, with the dislocation density tensor $\,\overline{{K}}\,$, since they are in a simple one-to-one relation.

\section{Integrability conditions for the symmetric stretch tensor $U$}\label{sec3}
In classical continuum mechanics, the deformation of a continuous medium can be characterized by the Green-Lagrange strain tensor, which is defined by  $E\,=\,\frac{1}{2}\,(F^TF-\id_3)\,=\,\frac{1}{2}\,\left(\nabla\varphi^T\nabla\varphi-\id_3\right).$ Conversely, if a positive definite symmetric matrix $C=\id_3+2\,E$ is given in terms of $x$, it would be interesting to know under what conditions the system
\begin{align}
	\nabla\varphi^T\nabla\varphi=C
\end{align} 
of first order partial differential equations is liable to be solved. In \cite{Vallee92,Ciarlet06,CiarletVallee07} the authors have shown proper compatibility conditions in the case of a classical continuum undergoing large deformations, which yields the compatibility relation
\begin{align}\label{eq:comp}
	(\Curl A)^T+\Adj A=0,
\end{align}
where $\Adj A$ is the adjugate of $A$ and with 
\[
 A=\frac1{\det U}\left(\Big(U(\Curl U)^T-\frac{1}{2}\tr(U(\Curl U)^T)\id_3\Big)U\right). 
\]
In addition (cf. \cite{Vallee92}), C. Vall{\'e}e gets the following useful result for the symmetric stretch tensor $U=C^{1/2}$.
\begin{theorem}\label{th:comp1}
Let $\Omega\subset\R^3$ be a domain. Suppose $\varphi\colon\Omega\longrightarrow\R^3$ is an twice continously differentiable deformation field of an elastic body with polar decomposition
\begin{align}
\nabla\varphi(x)=R(x)\,U(x),
\end{align}
where $R\colon\Omega\longrightarrow$ and $U\colon\Omega\longrightarrow\mathrm{Psym}(3)$ is positive definite, symmetric and differentiable. Then for any $x\in\Omega$ and $u,v\in\R^3$ we have 
\begin{align}\label{eq:twentytwo}
	\left[\big(R^T\Grad R\,\big)(x).u\right].v=\frac{1}{\det U(x)}\Big(\Big\{U(x)(\Curl U(x))^T-\frac{1}{2}\tr(U(x)(\Curl U(x))^T)\id\Big\}\,U(x)u\Big)\times v
\end{align}
\end{theorem}

Alternatively to \eqref{eq:twentytwo}, we have for $x\in\Omega$ 
\begin{align}\label{eq:twentyfive}
	\Axl \left(R^T\Grad R\right)(x)=\frac{1}{\det U(x)}\Big\{U(x)(\Curl U(x))^T-\frac{1}{2}\tr(U(x)(\Curl U(x))^T)\id\Big\}\,U(x).
\end{align}

\begin{remark}\label{rem3}
Let $\Omega\subset\R^3$ be a domain. Theorem \ref{th:comp1} is also well-known for the case of infinitesimal deformations $\varphi\colon\Omega\longrightarrow\R^3$ (see \cite{Ghiba_15}) in which the polar decomposition 
$\nabla\varphi=F=RU$ turns into
\begin{align}
\id+\nabla u=(\id+\mathrm{skew}\nabla u+\ldots)(\id+\sym\nabla u+\ldots)
\end{align}
in the neighbourhood of $\id$, where $\mathrm{skew}\nabla u$ is the infinitesimal continuum relation and $\sym\nabla u$ is the infinitesimal stretch. Thus
\begin{align}
\nabla u=\sym\nabla u+\mathrm{skew}\nabla u
\end{align}
is the linearized polar decomposition and it holds 
\begin{align}
	\nabla\axl(\mathrm{skew}\nabla u)&=\big[\Curl\sym\nabla u\big]^T-\frac{1}{2}\tr\Big(\big(\Curl\sym\nabla u\big)^T\Big)\id=\big[\Curl\sym\nabla u\big]^T,
\end{align}
which is the linearization analogue of \eqref{eq:twentyfive}.
\end{remark}

\section{Integrability conditions for the non-symmetric first Cosserat deformation tensor $\oU$}\label{sec4}
In \cite{Yavari13} the authors obtained a compatibility condition for the right Cauchy-Green tensor, using homology and homotopy group techniques and in \cite{Pietraszkiewicz83,Shield73} one can find compatibility conditions for rotations. In a more direct approach we get the following result similar to Theorem \ref{th:comp1}. 
\begin{theorem}\label{thm:main_one}
Let $\Omega\subset\R^3$ be a domain and $\varphi\colon \Omega\longrightarrow \R^3$ be a twice continuously differentiable function. Suppose, for any $x\in\Omega$, 
\begin{align} \label{eq:decomposition}
\nabla\varphi(x)=\oR(x)\,\oU(x), 
\end{align}
where $\oR\colon \Omega\longrightarrow \mathrm{O}(3)$ and $\oU\colon \Omega\longrightarrow \mathrm{GL}(3)$ are continuously differentiable. Then, for any $x\in \Omega$
\begin{equation}\label{eq:AxlRTGradR}
\mathrm{Axl}\,\big(\oR^T\Grad \oR\,\big)(x)=\frac{1}{\det \oU(x)}\Big\{\oU(x)(\Curl \oU(x))^T-\frac{1}{2}\tr(\oU(x)(\Curl \oU(x))^T)\id\Big\}\,\oU(x).
\end{equation}
\end{theorem}
\begin{remark}
Similar to Remark \ref{rem3} we obtain the formula 
\[
 \nabla \axl(\oA) = \Big(\Curl(\nabla u-\oA)\Big)^T-\frac12\tr\left(\Big(\Curl(\nabla u-\oA)\Big)^T\right)\cdot \id
\]
as linearization of \eqref{eq:AxlRTGradR}, where $\nabla u-\oA$ is the linearization of the first Cosserat deformation tensor and $\nabla \axl(\oA)$ is the linearization of the second Cosserat curvature tensor.
\end{remark}

As partial converse, we obtain formula \eqref{eq:comp} as condition for the existence of a deformation $\phii$ admitting the decomposition \eqref{eq:decomposition} also in the case of not necessarily symmetric tensors $\oU$:

\begin{theorem}\label{thm:main_two}
 Let $\Om\sub\R^3$ be a simply connected domain. Let $\oU\colon\Om\to\Rdd$ be continuously differentiable with values in $GL(3)$. Then there exists a deformation $\phii\col \Om\to \R^3$ and $\oR\col \Om\to O(3)$ satisfying 
  \[
   \nabla \phii(x)=\oR(x)\oU(x) \qquad \mbox{for all }x\in\Om,
  \]
provided that 
\begin{equation}\label{eq:thm_defA}
 A:= \frac1{\det \oU}\left(\oU(\Curl \oU)^T-\frac12\tr(\oU(\Curl \oU)^T)\id\right)\,\oU
\end{equation}
is continuously differentiable and satisfies the compatibility condition 
\[
 (\Curl A(x))^T + \Adj A(x) =0 \qquad \mbox{for every }x\in\Om.
\]
Furthermore this deformation is unique up to a rigid body deformation, that is multiplication with a constant rotation and addition of a constant vector.
\end{theorem}

The representation given by Theorem \ref{thm:main_one} makes it possible to express the Cosserat tensors in terms of $\oU$ and $\Curl\oU$ in the following way:

\begin{consequence}
Let $\Omega\subset\R^3$ be a domain and $\varphi\colon \Omega\longrightarrow \R^3$ be a twice continuously differentiable function. Suppose, we have the decomposition
\begin{align} \label{eq:31}
\nabla\varphi(x)=\oR(x)\,\oU(x) \qquad \mbox{for any } x\in \Om, 
\end{align}
where $\oR\colon \Omega\longrightarrow O(3)$ and $\oU\colon \Omega\longrightarrow GL(3)$ are differentiable. Then, we have
\begin{align}\label{eq_C_1}
 \mathfrak{K}&=\oR^T\Grad \oR=\Anti\left(\frac{\Big[\oU(x)(\Curl \oU(x))^T-\frac{1}{2}\tr(\oU(x)(\Curl \oU(x))^T)\id_3\Big]\oU(x)}{\det \oU}\right),\\
 \overline{K}&=-\mathfrak{K}:\epsilon=-\Anti \left(\frac{\Big[\oU(x)(\Curl \oU(x))^T-\frac{1}{2}\tr(\oU(x)(\Curl \oU(x))^T)\id_3\Big]\oU(x)}{\det \oU}\right):\epsilon.
\end{align}
\end{consequence}
\begin{proof}
Equation \eqref{eq:31} directly follows from the definition \eqref{eq:defAnti} of $\Anti$ and \eqref{eq:AxlRTGradR}.
Formula \eqref{eq_C_1} hence is a consequence of the identity $\oK=-\fraK:\epsilon$, which has been proven in \cite{NBO_2014}.
\end{proof}

Another consequence of Theorem \ref{thm:main_one} is that the $\Curl$ of $\oU$ can be written as function of $\oU$ and $\oK$ as follows. Note, however, that it is to be expected that the conditions concerning regularity and structure (orthogonality for $\oR$ and gradient-structure of $\oR\,\oU$ simultaneously) make extending this consequence to less regular functions more difficult than the usual trivial approximation by sequences of mollified functions. 

\begin{consequence}
Let $\Omega\subset\R^3$ be a domain and $\varphi\colon \Omega\longrightarrow \R^3$ be a twice continuously differentiable function. Suppose we have the decomposition
\begin{align}
\nabla\varphi(x)=\oR(x)\,\oU(x) \qquad\mbox{for any }x\in\Omega\,, 
\end{align}
where $\oR\colon \Omega\longrightarrow O(3)$ and $\oU\colon \Omega\longrightarrow GL(3)$ are continuosly differentiable. Then,  we have
\begin{align}
(\Curl \oU)^T&=\Adj\oU\left(\Gamma\oU\,^{-1}-\tr\Big(\Gamma\oU\,^{-1}\Big)\id_3\right),\quad\text{that is}\label{eq:val_lank}\\
(\Curl\oU)^T&=\Adj\oU\left[\frac{1}{2}\tr(\oK)\Big(\oU\,^{-1}-\tr\big(\oU\,^{-1}\big)\id_3\Big)-\left(\oK^T\oU\,^{-1}-\tr\left(\oK^T\oU\,^{-1}\right)\id_3\right)\right]. \label{eq:val_lank_2}
\end{align}
\end{consequence}

\begin{proof}
Again we define:
\[
A:=\frac{1}{\det\oU}\Big[\oU(\Curl \oU)^T-\frac{1}{2}\tr(\oU(\Curl \oU)^T)\id_3\Big]\oU,
\]
as in \eqref{eq:thm_defA}, so that \eqref{eq:Axl} readily shows 
\(
 \epsilon:\fraK = -2 \Axl \fraK = -2A
\)
according to \eqref{eq:AxlRTGradR}. Thus we have (using \eqref{e8})
\(
	\Gamma=-\frac{1}{2}\epsilon:\fraK=A
\)
and \eqref{eq:val_lank} follows from Lemma \ref{lem:DRU} below or, more precisely, from the first step of the proof of Theorem \ref{thm:main_one}, that is \eqref{eq:beginningthmproof}. 
Therefore also \eqref{eq:val_lank_2}, being the combination of \eqref{eq:val_lank} and \eqref{e11}, is proven.
\end{proof}

\subsection{Proof of the Theorem}
For the sake of notational convenience let us introduce the following notation in order to interpret the determinant of a matrix as function of its column vectors:
\begin{definition}\label{def:vol}
In the following, let $\vol\colon \R^3\times \R^3\times \R^3\longrightarrow \R$ denote the 3-linear alternating 3-form on $\R^3$ defined by 
\[
 \vol(u,v,w)=\det A, \qquad u,v,w\in\R^3
\]
where $A=(u|v|w)$.
\end{definition}
 
A direct consequence of the definition then is, that for $A\in\R^{3\times 3}$ and $u,v,w\in\R^3$ we have (cf. Appendix, Lemma \ref{lem:detadj})
\begin{align}
	\vol(Au,Av,Aw)=\det A\vol(u,v,w)\quad\text{ and }\quad\vol(Au,Av,w)=\vol(u,v,(\Adj A) w),\label{detadj}
\end{align}
where $\Adj A$ denotes the adjugate of $A$. 

If the matrix is applied to only one argument of $\vol$, we obtain the following formula (cf. Appendix, Lemma \ref{lem:tr}):
\begin{align}
\vol(Au,v,w)+\vol(u,Av,w)+\vol(u,v,Aw)=(\tr A) \vol(u,v,w) \label{tr}
\end{align}
for all $u,v,w\in\R^3$ and $A\in\R^{3\times 3}$.
\begin{lemma}
 \label{lem:commutevectorproducts}
 Let $A,B\in\Rdd$, $\det B\neq 0$, let $h,k\in\R^3$. Then 
\[
 (Ah)\times (Bk)- (Ak)\times(Bh) = \left[\tr(AB\inv) (\Adj B)^T - ((\Adj B)AB\inv)^T \right] h\times k.
\]
\end{lemma}
\begin{proof}
Let us consider the scalar product of the left-hand-side with arbitrary $v\in\R^3$. Since $\sprod{a\times b, v}=\vol(a,b,v)$ for all $a,b,v\in\R^3$ (cf. Appendix, Lemma \ref{lem:vectorproduct}) we have
\begin{align*}
 L:=\sprod{Ah\times Bk-Ak\times Bh, v}=&\vol(Ah,Bk,v)-\vol(Ak,Bh,v)\\
=&\vol(BB\inv Ah,Bk,v)-\vol(BB\inv Ak,Bh,v),
\end{align*}
which, by $\eqref{detadj}_2$ equals 
\[
 \vol(B\inv Ah,k,\Adj B v)-\vol(B\inv Ak,h,\Adj B v)
\]
and after switching two arguments of the second $\vol$ and adding and substracting \\$\vol(h,k,B\inv A\Adj Bv),$ we infer
\begin{align*}
 L=&\vol(B\inv Ah,k,\Adj Bv)+\vol(h,B\inv A k, \Adj B v)\\
      &\hspace{0.5cm}+\vol(h,k,B\inv A\Adj Bv)-\vol(h,k,B\inv A\Adj Bv).
\end{align*}
Here \eqref{tr} can be applied to combine the first three terms and since $\sprod{a\times b, v}=\vol(a,b,v)$
(see \eqref{vectorproduct}), we get 
\begin{align*}
 L=&\tr(B\inv A)\vol(h,k,\Adj B v)-\vol(h,k,B\inv A\Adj Bv)\\
 =&\tr(B\inv A)\sprod{h\times k,\Adj Bv}-\sprod{h\times k, B\inv A\Adj Bv}\\
 =&\tr(B\inv A)\sprod{\Adj B^T h\times k,v}-\sprod{(B\inv A\Adj B)^T h\times k,v}\\
 =&\sprod{\left[\tr(B\inv A)\Adj B^T-(B\inv A\Adj B)^T\right]h\times k,v}.
\end{align*}
Due to the arbitrariness of $v$, this proves the claim if we take into account that $\tr (B^{-1} A)=\tr (AB^{-1})$.
\end{proof}

\begin{lemma}\label{lem:baccab}
Let $A\in\Rdd$. Then, for any $h,k,v\in \R^3$,
\[
 ((\Adj A)^T(h\times k))\times v = Ah\times (Ak\times v) - (Ak\times Ah)\times v.
\]
\end{lemma}
\begin{proof}
Let us recall that, for any $a,b,c\in\R^3$, 
\(
 a\times(b\times c) = \sprod{a,c}b-\sprod{a,b}c.
\)
Using this rule twice, for arbitrary $h,k,v\in \R^3$ we obtain 
\begin{align*}
 Ah\times(Ak\times v)&-Ak\times(Ah\times v)= \sprod{Ah,v}Ak - \sprod{Ah,Ak}v-\sprod{Ak,v}Ah+\sprod{Ak,Ah}v 
 \\=&\sprod{Ah,v}Ak-\sprod{Ak,v}Ah=v\times(Ak\times Ah)=(Ah\times Ak)\times v=((\Adj A)^T(h\times k))\times v,
\end{align*}
where the last equality results from the fact that for any $h,k\in\R^3$ and $A\in\Rdd$, 
 \(
  Ah\times Ak = (\Adj A)^T (h\times k)
 \)
 due to the second part of \eqref{detadj}.
\end{proof}

\begin{lemma} \label{lem:DUhk}
Let $\oU\colon\Omega\longrightarrow \R^{3\times 3}$ be differentiable. Then for any $x\in \Omega$, $h,k\in\R^3$, we have 
\[
 (\Grad U(x).h)k-(\Grad U(x).k)h=\Curl U(x).(h\times k)\,.
\]
\end{lemma}

\begin{proof}
According to the Definition we have
\begin{align*}
	\mathrm{Curl}\,U(x).(h\times k)&=\big(-U,_i(x)\times e_i\big)(h_r\,e_r\times k_s\,e_s)
							 =\big(\varepsilon_{ijk}U_{mk,j}(x)\,e_m\otimes e_i\big)(\varepsilon_{rst}h_sk_t\,e_r)\\
						&=\varepsilon_{ijk}\varepsilon_{rst}U_{mk,j}(x)h_sk_t\delta_{ir}\,e_m
						   =\varepsilon_{ijk}\varepsilon_{ist}U_{mk,j}(x)h_sk_t\,e_m\\
						&=-\varepsilon_{jik}\varepsilon_{ist}U_{mk,j}(x)h_sk_t\,e_m
						   =\big(\delta_{js}\delta_{kt}-\delta_{jt}\delta_{ks}\big)U_{mk,j}(x)h_sk_t\,e_m\\
						&=\big(U_{mt,s}(x)-U_{ms,t}(x)\big)h_sk_t\,e_m
\end{align*}
and on the other hand
\begin{align*}
	\big(\Grad U(x).h\big)k-\big(\Grad U(x).k\big)h&=\big((U,_k(x)\otimes\,e_k)h_l\,e_l\big)k_r\,e_r-\big((U,_k(x)\otimes \,e_k)k_r\,e_r\big)h_l\,e_l\\
	&=\big(U_{ij,k}(x)h_l\delta_{lk}\,e_i\otimes e_j\big)k_r\,e_r-\big(U_{ij,k}(x)k_r\delta_{kr}\,e_i\otimes e_j\big)h_l\,e_l\\
	&=\big(U_{ij,l}(x)h_lk_r\delta_{jr}-U_{ij,r}(x)k_rh_l\delta_{jl}\big)\,e_i\\
	&=\big(U_{ir,l}(x)-U_{il,r}(x)\big)h_lk_r\,e_i
\end{align*}
and therefore the lemma is proved.
\end{proof}

\begin{lemma}
\label{lem:RTDRh}
Let $\oR\colon \Omega\longrightarrow \Rdd$ be a differentiable function with values in $O(3)$. Then for any $x\in \Omega$, $h\in\R^3$, 
\begin{align}\label{eq:skew}
	\fraK.h=\oR^T(x) \Grad \oR(x).h \in \R^{3\times 3} \quad \text{ is skew-symmetric}
\end{align}
and there is a map $A\colon \Omega\times \R^3\longrightarrow\R^3$, linear with respect to its second argument, such that for any $x\in\Omega$, $h\in\R^3$, $v\in\R^3$, 
\begin{align}\label{eq:DRhv}
	(\Grad \oR(x).h)v=\oR(x)A(x,h)\times v\,.
\end{align}
\end{lemma}

\begin{proof}
 Because $\oR(x)\in O(3)$ for any $x\in \Omega$, 
\[
 \id_3=\oR^T(x)\oR(x)
\]
holds true for any $x\in \Omega$ and differentiating leads to 
\[
 0=\big[\Grad \oR(x).h\big]^T \oR(x)+\oR(x)^T \Grad \oR(x).h = 2\sym(\oR(x)^T \Grad \oR(x).h)
\]
for arbitrary $x\in \Omega$, $h\in\R^3$, which shows that the symmetric part of $\oR(x)^T\,\Grad\,\oR(x).h$ vanishes and implies \eqref{eq:skew}.
By definition $\eqref{eq:not_anti}_3$, $\oR^T(x) \Grad \oR(x).h$ therefore can be expressed as 
\[
 (\oR^T(x)\Grad \oR(x).h)v=A(x,h)\times v,
\]
where $A(x,h)$ depends linearly on $h$. Now, equation \eqref{eq:DRhv} follows directly upon multiplication by $\oR(x)$.
\end{proof}

In order to assert solvability of differential equations, let us recall the Theorem of Frobenius (see \cite[Thm. 4.E, p. 167]{ZeidlerI}):
\begin{lemma}\label{lem:Frob}
 Let $X,Y$ be Banach spaces, $U\subset X\times Y$ open and $F\col U\to L(X,Y)$ a continuously differentiable function. Then the following are equivalent:\\
 i)For any $(x_0,y_0)\in U$, the differential equation $Dy(x)=F(x,y(x))$, $y(x_0)=y_0$ has exactly one $C^1$-solution in a neighbourhood of $(x_0,y_0)$.\\
 ii) $D_x F(x,y) hk +D_yF(x,y)F(x,y)hk = D_xF(x,y)kh + D_yF(x,y)F(x,y)kh$ for all $(x,y)\in U$ and all $h,k\in X$.
\end{lemma}
We immediately take this lemma to its use in the following:

\begin{lemma}\label{lem:solverot}
 Let $\Om\sub\R^3$ be simply connected. 
 Given a differentiable function $A\col \Om\to \Rdd$ 
 such that 
 \[
  (\Adj A(x))^T + \Curl A(x) =0 \qquad \mbox{for all }x\in \Om,
 \]
 there exists a continuously differentiable function $R\col \Om\to\Rdd$ with values in $O(3)$ such that 
 \begin{equation}\label{eq:diffeq}
  [\Grad R(x).h]v=R(x) ((A(x)h)\times v) \qquad \mbox{for all }x\in\Om,\,h,v\in\R^3.
 \end{equation}
 This function is unique up to the multiplication by a constant $Q\in O(3)$.
\end{lemma}
\begin{proof}
 We want to apply Frobenius' theorem (Lemma \ref{lem:Frob}) with $F\col \Om\times \Rdd\to L(\R^3,\Rdd)$ defined by 
 \[
  [F(x,Y)h]v = Y[A(x)h]v, \qquad x\in \Om, Y\in \Rdd, h,v\in \R^3.
 \]
 Thus, the partial derivatives of $F$ are given by 
 \begin{align*}
  [D_xF(x,Y)hk ]v =& Y[D(A(x)k)h]\times v,\\
   [[D_YF(x,R)Z]\xi]v=&Z(A(x)\xi)v,\\
   [D_YF(x,Y)F(x,Y)hk]v=&F(x,Y)h(A(x)k)\times v=Y(A(x)h)\times((A(x)k)\times v).
 \end{align*}
 Condition 2 from Frobenius' theorem (Lemma \ref{lem:Frob}) therefore is
 \begin{equation}\label{eq:frobcond}
 Ah\times (Ak\times v) - Ak\times(Ah\times v) + (\Curl A(h\times k))\times v=0\qquad \forall h,k,v, 
\end{equation}
which we can obtain as follows:
\[
 (\Adj A)^T + \Curl A =0
\]
implies that 
\[
 (((\Adj A)^T +\Curl A)(h\times k))\times v = 0 \qquad \mbox{for all }v\in \R^3, h,k\in \R^3.
\]
By Lemma \ref{lem:baccab}, this entails \\
\[
 Ah\times (Ak\times v) - Ak\times(Ah\times v) + (\Curl A(h\times k))\times v=0\qquad \forall h,k,v,
\]
which is \eqref{eq:frobcond}.
Therefore, since $\Om$ is simply connected, choosing an arbitrary point $x_0\in\Om$ and $y_0\in O(3)$ yields a unique solution $R$ to \eqref{eq:diffeq} satisfying $R(x_0)=y_0$. 
Since $R^T DR(x)h$ is skew-symmetric for any $x\in \Om$ and any $h\in\R^3$ by \eqref{eq:diffeq} (it can be expressed as vector product), $R^TR$ is constant within $\Om$ and hence $R^T(x)R(x)=y_0^Ty_0=\id$, so that actually $R(x)\in O(3)$ for all $x\in \Om$.
Given a second solution $\Rtilde\col \Om\to O(3)$, by the aforementioned uniqueness property, $R(x_0)\Rtilde(x_0)^{-1}\Rtilde=R$.
\end{proof}

\textbf{Proof of the theorems.} 
The theorems can be seen as partially converse statements. Let us begin their proofs with the part they have in common:
\begin{lemma}\label{lem:DRU}
 Let $\oU\col \Om\to \Rdd$ be differentiable and $\oU(x)$ invertible for every $x\in\Om$ and let $\oR\col \Om\to \Rdd$ be a differentiable function with values in $O(3)$. Then 
 \begin{equation}\label{eq:GradRUsym}
  \Grad(\oR\,\oU)(x)(k,h)=\Grad(\oR\,\oU)(x)(h,k)
 \end{equation}
 if and only if 
\begin{equation}\label{eq:AUCurl}
 \det \oU(x)[A(x)\oU\inv(x)-\tr(A(x)\oU\inv(x))\id_3]=\oU(x)(\Curl \oU(x))^T \quad \text{for any }x\in\Om,
\end{equation}
  where $A$ is defined as in Lemma \eqref{eq:DRhv}.
\end{lemma}
\begin{proof}
The product rule implies that
 \[
  \Grad(\oR\,\oU)(x)(h,k)=[\Grad \oR(x).k]\oU(x)h+\oR(x)[\Grad \oU(x).k]h \quad \text{for any }x\in \Om, h,k\in\R^3
 \]
 and calculating the difference between two such expressions with different order of the arguments, we obtain 
\begin{align*}
\Grad(\oR\,\oU)(x)(h,k)&-\Grad(\oR\,\oU)(x)(k,h)\\
=& (\Grad \oR(x).k)\oU(x)h+\oR(x)[\Grad \oU(x).k]h-[\Grad \oR(x).h]\oU(x)k-\oR(x)[\Grad \oU(x).h]k.
\end{align*}
An application of Lemma \ref{lem:RTDRh} shows that \eqref{eq:GradRUsym} is equivalent to 
\[
 0=\oR(x)A(x,k)\times(\oU(x)h)+\oR(x)[\Grad \oU(x).k]h-\oR(x)A(x,h)\times(\oU(x)k)-\oR(x)[\Grad \oU(x).h]k
\]
for any $x\in \Om$, $h,k\in\R^3$, where $A(x,h)$ is defined as in \eqref{eq:DRhv}.
Multiplying this equation by $\oR^T(x)$ and re-ordering the terms, we equivalently arrive at 
\[
 A(x)h\times \oU(x)k- A(x)k\times \oU(x)h = -\big([\Grad \oU(x)h]k+[\Grad \oU(x)k]h\big)
\]
for any $x\in\Om$, $h,k\in\R^3$, where we can apply the equalities from Lemmata \ref{lem:commutevectorproducts} and \ref{lem:DUhk} to the left-hand-side and right-hand-side respectively.
This transforms the equation into
\[
 [\tr\left(A(x)\oU\inv(x)\right)[\Adj \oU(x)]^T-((\Adj \oU(x))A(x)\oU\inv(x))^T] h\times k = - \Curl \oU(x) (h\times k)
\]
for all $x\in\Om$, $h,k\in\R^3$ and therefore, by transposition and because $h,k$ were arbitrary, 
\[
 \Adj \oU(x)[A(x)\oU\inv(x)-\tr(A(x)\oU\inv(x))\id_3] = (\Curl \oU(x))^T
\]
for any $x\in \Om$. Multiplication by the invertible matrix $\oU(x)$ from the left gives
\begin{equation*}
 \det \oU(x)[A(x)\oU\inv(x)-\tr(A(x)\oU\inv(x))\id_3]=\oU(x)(\Curl \oU(x))^T \quad \text{for any }x\in\Om\,,
\end{equation*}
which finally asserts that \eqref{eq:GradRUsym} and \eqref{eq:AUCurl} are equivalent.
\end{proof}

We continue with the proof of Theorem \ref{thm:main_one}.

\begin{proof}[Proof of Theorem \ref{thm:main_one}]
 As $\nabla\phii=\oR\,\oU$, by Schwarz' theorem and Lemma \ref{lem:DRU}, we obtain $0=D^2\phii(x)hk-D^2\phii(x)kh$ and hence 
 \begin{equation}\label{eq:beginningthmproof}
  \det \oU(x)[A(x)\oU\inv(x)-\tr(A(x)\oU\inv(x))\id_3]=\oU(x)(\Curl \oU(x))^T \quad \text{for any }x\in\Om.
 \end{equation}
Taking the trace on both sides, we infer 
\[
 \det \oU(x)[\tr(A(x)\oU\inv(x))-\tr(A(x)\oU\inv(x))\cdot 3] = \tr(\oU(x)(\Curl \oU(x))^T)
\]
for any $x\in \Om$ and hence 
\[
 \tr \left(A(x)\oU\inv(x)\right)=-\frac1{2\det \oU(x)} \tr(\oU(x)(\Curl \oU(x))^T) \qquad \text{for all  }x\in\Om.
\]
Inserting this into \eqref{eq:AUCurl} shows 
\[
 \det \oU(x)[A(x)\oU\inv(x)+\frac{1}{2\det \oU(x)}\tr(\oU(x)(\Curl \oU(x))^T)\id_3]=\oU(x)(\Curl \oU(x))^T\quad \text{for }x\in \Om,
\]
which entails 
\[
 A(x)=\frac1{\det \oU(x)}\oU(x)(\Curl \oU(x))^T \oU(x)-\frac1{2\det \oU(x)}\tr[\oU(x)(\Curl \oU(x))^T] \oU(x) \quad \text{for all } x\in\Om.
\]
By definition \eqref{eq:DRhv} of $A$ finally this translates into 
\[
 [\Grad \oR(x).h]v= \oR(x)[(\frac1{\det \oU(x)}\oU(x)(\Curl \oU(x))^T \oU(x)-\frac1{2\det \oU(x)}\tr[\oU(x)(\Curl \oU(x))^T] \oU(x)) h]\times v
\]
for any $x\in \Om$, $h,v\in\R^3$.
\end{proof}

\begin{proof}[Proof of Theorem \ref{thm:main_two}]
Since $A$ satisfies the compatibility condition, Lemma \ref{lem:solverot} enables us to find $\oR\col \Om\to O(3)$ such that 
\[
 [\Grad \oR(x).h]v=\oR(x) ((A(x)h)\times v) \qquad \mbox{for all }x\in\Om,\,h,v\in\R^3.
\]
Furthermore, $\oR$ is determined uniquely up to multiplication with a constant rotation.\\
The definition of $A$ ensures that 
\[
 A\,\oU\inv = \frac{1}{\det \oU}(\oU(\Curl \oU)^T - \frac12 \tr (\oU(\Curl \oU)^T)\id)
\]
and hence 
\begin{align*}
 \det \oU(A \oU\inv-\tr(A\oU\inv)\id)=&\oU(\Curl \oU)^T-\frac12\tr(\oU(\Curl \oU)^T)\id-\tr\left[\oU(\Curl \oU)^T - \frac12\tr(\oU(\Curl \oU)^T)\id\right]\id\\
 =&\oU(\Curl \oU)^T-\frac12\tr(\oU(\Curl \oU)^T)\id-\tr(\oU(\Curl \oU)^T)\id+\frac32\tr(\oU(\Curl \oU)^T)\id\\
 =&\oU(\Curl \oU)^T
\end{align*}
that is, condition \eqref{eq:AUCurl} from Lemma \ref{lem:DRU} is met, so that condition 2 of Lemma \ref{lem:Frob} is satisfied for $F(x,Y)=\oR(x)\oU(x)$ and hence upon choosing $x_0\in \Om$ and $\phii(x_0)\in\R^3$ there is a unique mapping $\phii\col\Om\to\R^3$ such that $\nabla \phii=\oR\,\oU$.
\end{proof}

\section{Existence theorem for a new Cosserat formulation}\label{sec5}
We have seen that in fact $\oK$ can entirely be obtained as a complicated nonlinear function of the strain measure $\oU$ and its derivative in the form $\Curl \oU$ alone; therefore, any formulation directly written as a simple function of $(\oU, \oK)$ is a function of $(\oU, \Curl \oU)$, in which, however, the effects of strain $\oU$, and generalized curvature $\Curl \oU$ are not separated. Next, we would like to present an energetic formulation which has the possibility to additively separate the pure straining part $\sym(\oU-\id)$  and $\Curl \oU$ for certain parameter ranges. 
Such a formulation has not been considered previously. We set 
\begin{align}
 I(\varphi,\oR)=& \int_\Omega W(\oU,\oK)\, \mathrm{d}V -\Pi(\varphi, \oU),\qquad \mathrm{where} \label{eq:minprob}\\
\label{eq:energy}
 W(\varphi,\oR)=&W(\oU,\oK)=\mu\norm{\sym(\oU-\id)}^2+\mu_c\norm{\mathrm{skew}(\oU-\id)}^2\nonumber \\
& \qquad \qquad \quad +\frac{\lambda}{2}(\det\oU^2-\frac{1}{\det\oU^2}-2)+\mu L_c^2\norm{\Curl(\oU-\id)}^2+\mu \widehat{L}_c\norm{\oK}^4,\\
\Pi(\varphi,\oU)=& \int_\Omega f\, \varphi\, dV\notag 
\end{align}
for a given body force $0\not\equiv f\in L^2(\Omega)$. 
Taking into account \eqref{eq:val_lank} and \eqref{eq:val_lank_2} it is clear that, for sufficiently smooth functions,
\(
\norm{\Curl(\oU-\id)}^2=\norm{\Curl \oU}^2
\)  
is indeed a complicated, nonlinear function of $\oU$ and $\oK$. \\

In \eqref{eq:energy}, $\mu_c=0$ means that only the symmetric part of the straining is taken into account and $\widehat{L}_c$ is an additional length scale which makes \eqref{eq:minprob} well-posed. 
It would be the final goal to show that \eqref{eq:minprob} 
with $W(\varphi,\oR)$ given by \eqref{eq:energy} admits minimizers for the limit parameter range $\mu_c=0$ and $\widehat{L}_c=0$ in a suitable Soboloev-space dictated by the coercivity of the problem. 
Unfortunately, currently we cannot dispose of the contribution of $\norm{\oK}$ by setting $\widehat{L}_c=0$ and thereby completing the separation of the effects of $\oU$ and those of $\Curl \oU$. 
Let us summarize what can be obtained for the minimization problem posed by \eqref{eq:minprob}:
\begin{enumerate}
	\item[]Case (I): $\mu>0,\,\lambda>0,\,L_c>0;\qquad \mu_c>0,\,\widehat{L}_c>0,\,q\geq 2$ \\[6pt]
		$\phantom{Case (I): }\hookrightarrow$ existence of minimizers along the lines presented in \cite{NBO_2014}.
	\item[]Case (II): $\mu>0,\,\lambda>0,\,L_c\ge 0;\qquad \mu_c=0,\,\widehat{L}_c>0,\,q>3$\\[6pt]
		$\phantom{Case (II:}\hookrightarrow$ existence of minimizers along the lines presented in \cite{NeffCMT} using \cite{NeffPompe_Korn}.
	\item[]Case (III): $\mu>0,\,\lambda>0,\,L_c>0;\qquad \mu_c=0,\,\widehat{L}_c>0,\,q\geq 2$\\[6pt]
	         $\phantom{Case (III\,:}\hookrightarrow$ outlined subsequently
	\item[]Case (IV): $\mu>0,\,\lambda>0,\,L_c>0;\qquad\mu_c=0,\,\widehat{L}_c=0$\\[6pt]
		$\phantom{Case (IV): }\hookrightarrow$ yet open
\end{enumerate}

We assume fixed boundary conditions of place at $\Gamma_D\subset\partial\Omega$, i.e. $\varphi(x)\ein_{_{\Gamma_D}}=x$. Consistent with this requirement we impose zero tangential straining at $\Gamma_D$ which means (taking the cross product of a matrix with a vector row wise)
\begin{align*}
	&(\oU-\id)\times \normal =0,\quad \text{where $\normal$ denotes the outer unit normal}\\
	&\;\;\;\Longleftrightarrow\;\;(\oU-\id).\tau=0\quad\text{ for all tangential directions } \tau \text{ at }\Gamma_D.
\end{align*}
The latter condition translates into
\begin{align}
	\oR^T\nabla\varphi.\tau=\tau \text{ at $\Gamma_D$}
\end{align}
and due to $\varphi(x)\ein_{_{\Gamma_D}}=x$ we have $\nabla\varphi(x).\tau=\tau$. Thus we require in fact
\begin{align}
	\oR.\tau\ein_{_{\Gamma_D}}=\tau\quad\Longleftrightarrow\quad\oR\ein_{_{\Gamma_D}}=\id.
\end{align}
Since the body force $f$ is non-vanishing, the solution $(\varphi,\oR)$ will be non-trivial. \\

\subsection{Outline for Case (III)}
We use the direct methods of the calculus of variations and only indicate those arguments which are needed in addition to the proof in \cite{NBO_2014} for case (I). At first, let us recall the incompatible Korn's inequality (cf. \cite{Starke13,Witsch15}):
\begin{align}\label{eq:Staff}
	\int_{\Omega}\norm{\sym P}^2+\norm{\Curl P}^2\,dV\geq c^+(\Omega)\cdot\int_{\Omega}\norm{P}^2\,dV
\end{align}
Applying \eqref{eq:Staff} to $P=\oU-\id$, we obtain a uniform $L^2$-bound on $\oU-\id$ (while only controlling the symmetric part locally). Considering infimizing sequences $\varphi_k$, $\oR_k$ we obtain (due to the fact that $\oU-\id$ is uniformly $L^2$-bounded), similar to \cite{NBO_2014}, eq. (42):
\begin{align}
	\varphi_k\xrightharpoonup{H^1(\Omega)}\varphi\quad\text{ and }\quad\varphi_k\xrightarrow{L^2(\Omega)}\varphi, 
\end{align}
and
\begin{align}
	\Curl\oU_k\xrightharpoonup{H(\Curl)}\Curl\oU\quad\text{ and }\quad \oR_k^T\nabla\varphi_k=\oU_k\xrightharpoonup{L^2(\Omega)}\oU.
\end{align}
The independent control of $\oR_k$ is provided by that of $\norm{K}^4$, which gives 
\begin{align}
	\oR_k\xrightarrow{L^2(\Omega)}\oR\quad\text{ and }\quad \Curl\oR_k\xrightharpoonup{L^2(\Omega)}\Curl\oR.
\end{align}
On passing to a subsequence, not relabelled, we obtain
\begin{align*}
	\varphi_k&\xrightarrow{L^2(\Omega)}\varphi,\\
	\nabla\varphi_k&\xrightharpoonup{L^2(\Omega)}\nabla\varphi,\\
	\oU_k&\xrightharpoonup{L^2(\Omega)}\oU=\oR^T\nabla\varphi,\\
	\Curl\oU_k&\xrightharpoonup{L^2(\Omega)}\Curl\oU,\\
	\Curl\oR_k&\xrightharpoonup{L^2(\Omega)}\Curl\oR.
\end{align*}
Moreover, 
\[\Big(\det\oU^2+(\det\oU^2)^{-1}-2\Big)=(\det\nabla\varphi)^2+\frac1{(\det\nabla\varphi)^2}-2,\]
since $\det(\oR^T\nabla\varphi)=\det\nabla\varphi$. In addition we note that 
\(
 \xi\mapsto(\xi^2+\frac1{\xi^2}-2) 
\) 
is convex, hence 
\[
 \Big((\det\nabla\varphi)^2+\frac1{(\det\nabla\varphi)^2}-2\Big)
\]
is polyconvex and thus weakly lower semicontinuous. 
Concerning the other summands in the energy \eqref{eq:energy}, we observe that $\norm{\sym(\oU-\id)}^2$ is convex in $\oU$ and both curvature contributions are quadratic. Hence, 
\begin{align}
\int_\Omega W(\varphi,\oR)\,dV\leq\liminf_{k\rightarrow\infty}\int_{\Omega}W(\varphi_k,\oR_k)\,dV
\end{align}
and the proof is complete.
\begin{remark}
For the general case regarding boundary conditions, i.e. $\varphi_{|_{\Gamma_D}}=\varphi_0$ we modify the energy functional to include a boundary contribution which realizes a weak coupling between microrotations and deformation gradient. We add
\begin{align}
\int_{\Gamma_D}\|(\oU-\id)\times\normal\|^2\,dS,\label{eq:new}
\end{align}
where again the cross product of a matrix with a vector is taken row wise. Then the NPW-Korn estimate \eqref{eq:Staff} holds with appropriate changes including the boundary contribution and the existence proof carries over. Note that
\begin{align}
\int_{\Gamma_D}\|(\oU-\id)\times\normal\|^2\,dS=0\Longleftrightarrow (\oU-\id)\times\normal=0 \text{ a.e.}\\
\intertext{which implies}
								  (\oU-\id)\tau=(\oR^T\nabla\varphi-\id).\tau=0\quad\text{for all tangential directions }\tau.
\end{align}
The advantage of \eqref{eq:new} is that we do not presuppose a direct knowledge about the coupling of microrotations and deformations at the boundary $\Gamma_D\subset \partial\Omega$. Moreover, it is in general not possible to prescribe 
\[
(\oR^T\nabla\varphi-\id).\tau=0 \quad\text{for all tangential directions }\tau\text{ at }\Gamma_D,
\]
since the condition $\nabla\varphi.\tau=\oR.\tau$ cannot be met if $\nabla \varphi$ is length changing on $\Gamma_D$.
\end{remark}

\section{Acknowledgement}
The basic ideas of this paper, which was intended to become a joint article, have been discussed with Claude Vall\'ee. He passed away in 2014, before the draft of the article was finished.

{\footnotesize
\section{Appendix}
As mentioned in the notation section, the vector product $x\times y$ of two vectors $x,y\in \R^3$ has a close connection with skew-symmetric matrices.
For $v\in\R^3$ fixed, we have 
\[
 v\times w=\mathrm{anti}(v).w=\matr{0&-v_3&v_2\\v_3&0&-v_1\\-v_2&v_1&0\\} w\qquad\forall\,w\in\R^3
\]
and for any skew-symmetric $A\in\Rdd$, there is $v\in\R^3$ such that $Aw=v\times w$. A more geometrical interpretation of the vector product is given by its characterization ``the vector product of $a,b\in\R^3$ is perpendicular to $a$ and $b$ both and its length is given by the area of the parallelogram spanned by $a$ and $b$'', which implies that the volume of a parallelepiped with sides $a, b ,c$ is given by $\sprod{a\times b,c}$, as expressed by the following lemma:
\begin{lemma}
\label{lem:vectorproduct}
 Let $a,b\in\R^3$. Then $a\times b$ is the unique element of $\R^3$ which satisfies 
\begin{align}
  \sprod{a\times b, v}=\vol(a,b,v)\qquad\text{for all }v\in\R^3 \label{vectorproduct}
\end{align}
\end{lemma}
\begin{proof}
Instead of giving componentwise calculations let us recall the elementary geometric definition of the vector product: ``The vector product of $a,b\in\R^3$ is perpendicular to $a$ and $b$ both and its length is given by the area of the parallelogram spanned by $a$ and $b$'', which implies that the volume of a parallelepiped with sides $a, b ,c$ is given by $\sprod{a\times b,c}$.
\end{proof}
\begin{lemma}\label{lem:detadj}
Let $u,v\in\R^3$ and $A\in\R^{3\times 3}$ with adjugate $\Adj A$. Then we have
\begin{align*}
	\vol(Au,Av,Aw)&=\det A\vol(u,v,w),\\
	\vol(Au,Av,w)&=\vol(u,v,\Adj A w).
\end{align*}
\end{lemma}

\begin{proof}
Let $u,v,w\in\R^3$ and $A\in\R^{3\times 3}$. According to the definition of $\vol$ (Definition  \ref{def:vol}) we have
\begin{align*}
	\vol(Au,Av,Aw)&=\det(Au|Av|Aw)=\det\Big[A(u|v|w)\Big]\\
	                          &=(\det A)\det(u|v|w)=(\det A) \vol(u,v,w).
\end{align*}
Using this identity, the multilinearity of $\vol$ and the well known identity $A\,\Adj A=\det A\,\id$ we get
\begin{align*}
 \vol(Au,Av,w)&=\frac{1}{\det A} \vol(Au,Av,\det A\,w)\\
		      &=\frac{1}{\det A}\vol(Au,Av,A \Adj A w)=\frac{\det A}{\det A}\vol(u,v,\Adj A w).\qedhere
\end{align*}
\end{proof}

\begin{lemma}\label{lem:tr}
Let $u,v,w\in\R^3$ and $A\in\R^{3\times 3}$. Then
\[
	\vol(Au,v,w)+\vol(u,Av,w)+\vol(u,v,Aw)=(\tr A)\,\vol(u,v,w).
\]
\end{lemma}
\begin{proof}
\begin{align*}
&\vol(Au,v,w)+\vol(u,Av,w)+\vol(u,v,Aw)\\
&\hspace{0.5cm}=\det\matr{a_{11}u_1+a_{12}u_2+a_{13}u_3&v_1&w_1\\a_{21}u_1+a_{22}u_2+a_{23}u_3&v_2&w_2\\a_{31}u_1+a_{32}u_2+a_{33}u_3&v_3&w_3}+ \det \matr{u_1&a_{11}v_1+a_{12}v_2+a_{13}v_3&w_1\\u_2&a_{21}v_1+a_{22}v_2+a_{23}v_3&w_2\\u_3&a_{31}v_1+a_{32}v_2+a_{33}v_3&w_3}\\
&\hspace{1cm}+\det\matr{u_1&v_1&a_{11}w_1+a_{12}w_2+a_{13}w_3\\u_2&v_2&a_{21}w_1+a_{22}w_2+a_{23}w_3\\u_3&v_3&a_{31}w_1+a_{32}w_2+a_{33}w_3}\\
&=a_{11}u_1v_2w_3+a_{12}u_2v_2w_3+a_{13}u_3v_2w_3+a_{21}u_1v_3w_1\\
&\hspace{1.5cm}+a_{22}u_2v_3w_1+a_{23}u_3v_3w_1+a_{31}u_1v_1w_2+a_{32}u_2v_1w_2+a_{33}u_3v_1w_2\\
&\hspace{0.5cm}-a_{11}u_1v_3w_2-a_{12}u_2v_3w_2-a_{13}u_3v_3w_2-a_{21}u_1v_1w_3\\
&\hspace{1.5cm}-a_{22}u_2v_1w_3-a_{23}u_3v_1w_3-a_{31}u_1v_2w_1-a_{32}u_2v_2w_1-a_{33}u_3v_2w_1\\
&\hspace{0.5cm}+a_{11}u_3v_1w_2+a_{12}u_3v_2w_2+a_{13}u_3v_3w_2+a_{21}u_1v_1w_3\\
&\hspace{1.5cm}+a_{22}u_1v_2w_3+a_{23}u_1v_3w_3+a_{31}u_2v_1w_1+a_{32}u_2v_2w_1+a_{33}u_2v_3w_1\\
&\hspace{0.5cm}-a_{11}u_2v_1w_3-a_{12}u_2v_2w_3-a_{13}u_2v_3w_3-a_{21}u_3v_1w_1\\
&\hspace{1.5cm}-a_{22}u_3v_2w_1-a_{23}u_3v_3w_1-a_{31}u_1v_1w_2-a_{32}u_1v_2w_2-a_{33}u_1v_3w_2\\
&\hspace{0.5cm}+a_{11}u_2v_3w_1+a_{12}u_2v_3w_2+a_{13}u_2v_3w_3+a_{21}u_3v_1w_1\\
&\hspace{1.5cm}+a_{22}u_3v_1w_2+a_{23}u_3v_1w_3+a_{31}u_1v_2w_1+a_{32}u_1v_2w_2+a_{33}u_1v_2w_3\\
&\hspace{0.5cm}-a_{11}u_3v_2w_1-a_{12}u_3v_2w_2-a_{13}u_3v_2w_3-a_{21}u_1v_3w_1\\
&\hspace{1.5cm}-a_{22}u_1v_3w_2-a_{23}u_1v_3w_3-a_{31}u_2v_1w_1-a_{32}u_2v_1w_2-a_{33}u_2v_1w_3\\
&=a_{11}(u_1v_2w_3-u_1v_3w_2+u_3v_1w_2-u_2v_1w_3+u_2v_3w_1-u_3v_2w_1)\\
&\hspace{0.75cm}+a_{22}(u_2v_3w_1-u_2v_1w_3+u_1v_2w_3-u_3v_2w_1+u_3v_1w_2-u_1v_3w_2)\\
&\hspace{1.5cm}+a_{33}(u_3v_1w_2-u_3v_2w_1+u_2v_3w_1-u_1v_3w_2+u_1v_2w_3-u_2v_1w_3)\\
&=(\tr A)\,\det(u\,|v\,|w)\qedhere
\end{align*}
\end{proof}
}

\footnotesize{\bibliographystyle{abbrv}
\bibliography{literatur_Osterbrink}}

\end{document}